\documentclass[format=acmsmall, review=false]{acmart}

\usepackage{acm-ec-19}
\usepackage{booktabs} 

\newcommand{\be}{\begin{equation}}
\newcommand{\ee}{\end{equation}}
\newcommand{\meanover}[2]{E_{#1}{#2}}
\newcommand{\se}{\{e\}}
\newcommand{\si}{\{i\}}

\begin{document}
\title[Data Dividend]{Computing a Data Dividend}
\author{Eric Bax}

\begin{abstract}
Quality data is a fundamental contributor to success in statistics and machine learning. If a statistical assessment or machine learning leads to decisions that create value, then data contributors may want a share of that value. This paper presents methods to assess the value of individual data samples, and of sets of samples, to apportion value among different data contributors. We use Shapley values for individual samples and Owen values for combined samples, and we show that these values can be computed in polynomial time in spite of their definitions having numbers of terms that are exponential in the number of samples. 
\end{abstract}

\maketitle

\pagebreak


\section{Introduction}
Many organizations utilize data about individuals to perform various functions. In many cases, this involves analyzing data that is aggregated over many people, to make decisions that lead to actions that generate profits or losses. Search, email, and social media providers' use of their users' data to select advertisements to show their users \cite{edelman07,varian09,levin11} (often in exchange for providing their services free of charge) has prompted calls for users to be paid a \textit{data dividend} \cite{au-yeung19} -- a portion of the profits generated using their data. This paper details methods to compute a data dividend, either on a sample-by-sample basis or for a collection of data that forms part of the data set used to make decisions. 

The methods presented in this paper apply to data dividends for profits derived from advertising to users and also to data dividends for other ways of generating profits from data. Search, email, and social media providers can, and in some cases do, generate profits from data through methods other than advertising. Examples include using aggregate email data, such as emailed receipts for purchases, as input to buy-sell decisions in equity markets \cite{alt_data,quandl_email_receipts,quandl_email_gopro} or for market research \cite{kooti17,kooti17b} that can support such decisions. Similarly, aggregated search and social media data can be analyzed to make effective economic forecasts \cite{hellerstein12,damuri17,onder15}, which can also support buy-sell decisions, leading to profits or losses. 

Other types of organizations also derive benefits from aggregated data. Insurance companies use aggregated data to set rates, and medical organizations use it (often on a smaller scale) for clinical trials and to predict disease outbreaks (a function that can benefit from search, email, and social media data as well \cite{sharpe16,alessa18}). Many retailers use customer relationship management (CRM) systems to automatically determine which offers to extend to which customers, based on analysis of data aggregated over customers. In Iceland, individuals' medical, genetic, and genealogical data have been aggregated for use in development of medical products as well as for anthropological discoveries \cite{wired_icelanders,gudbjar15,ed15,ebenes18}. 

Organizations may argue that individuals or groups who provide data are already compensated adequately, because they receive services in exchange for their data. However, increasingly, users are becoming concerned with data privacy, as evidenced by legal changes such as the GDPR \cite{gdpr16} in the European Union, congressional hearings on how internet-based companies use user data in the United States, and even comments by Tim Cook, CEO of Apple, on behalf of users of Apple devices, in favor of giving users more control over and awareness about how their data is used \cite{cook18}. As users learn more about privacy issues, they are beginning to understand that their data drives value, and they may soon begin to insist on receiving a share of that value. 

The methods to compute data dividends presented in this paper are based on the idea that some data may be more valuable than other data for making decisions that drive profits. So the methods are based on concepts from economic mechanism design \cite{mcmillan03,nisan07,krishna02,milgrom05} that were designed to reward members or subgroups of coalitions for the value that they contribute to their coalitions. The concepts include Shapley values \cite{shapley} to compute values for individual samples of data and Owen values \cite{owen77} to compute the share of gains or losses to allocate to each contributor of multiple samples.

Sections \ref{review_shapley} and \ref{review_perm} offer brief reviews of Shapley and Owen values and of some basic results regarding permutations that we will use for derviations in later sections. Section \ref{freq_based} describes methods to compute data dividends for decision processes that rely on frequency-based analysis, for example deciding whether to take an action based on average outcomes for similar actions under similar circumstances in the past. Section \ref{knn} presents methods to compute data dividends for decision processes that use nearest-neighbor models for analysis -- models that employ voting over samples to decide which action to take. Section \ref{discussion} concludes with a discussion of potential areas of interest for future work. 


\section{Review: Shapley Values and Owen Values} \label{review_shapley}
We will use Shapley values \cite{shapley} to value the contribution of individual data examples and Owen values \cite{owen77} to value the contribution of sets of examples. So we first review these values. Both involve averages over a number of terms that is exponential in the number of participants, but our methods will collect terms so that computation is feasible. 

Imagine that a group of participants will work together to generate some value. How should we reward each of them? One idea is to award them in the order that they agree to join the effort, and give them the marginal value they generate by joining. However, this may lead to an impasse. For example, if we value a completed jigsaw puzzle but have no value for a partially completed one, then we would only award whoever places the last piece. So no one would have an incentive to put the first two pieces together, or to do any work on the puzzle to get to having just one piece left to fit. Even if they did, everyone would be surreptitiously hiding pieces in an attempt to have the last one. If you have had this happen to you, then you will understand the problem. 

To avoid that type of problem, Shapley's insight was to derive the value for a participant by averaging over all possible orders in which the participants could have decided to participate. In each ordering, we assess the marginal contribution of the participant should they decide to join after the participants that are before them in the ordering and before the participants that are after them in the ordering. In the case of doing a puzzle together, each piece comes last in an equal number of the permutations of puzzle pieces: for $n$ pieces, $(n-1)!$, which is $\frac{1}{n}$ of the $n!$ permutations. So each piece has an equal Shapley value, which is the value of completing the puzzle divided by the number of pieces. 

A general equation for the Shapley value for participant $i$ is:
\be
\meanover{\sigma \in P}{\left[v(S_i(\sigma) \cup \si) - v(S_i(\sigma))\right]}, \label{shapley_def}
\ee
where $P$ is the set of permutations of $1, \ldots, n$ for $n$ participants, $\sigma$ denotes permutations, $S_i(\sigma)$ is the set of entries before $i$ in $\sigma$, and $v(S)$ is the value if the participants indexed by $S$ all participate, and the others do not.

Sometimes groups of participants form coalitions, and each coalitions decides as a group whether to participate. Owen values address this challenge by averaging over permutations of coalitions and, within the participant's coalition, over fellow coalition members. This gives a value for each participant, and the value for a coalition is the sum of its participants' values. Let $m$ be the number of coalitions, and let $C_1, \ldots, C_m$ be the sets of participant indices in the coalitions. Then the Owen value for participant $i$ in coalition $h$ is:
\be
\meanover{\sigma_C \in P_C}{ \meanover{\sigma_h \in P_h}{\left[v\left(\cup_{j \in S_h(\sigma_C)} C_j \cup S_i(\sigma_h) \cup \si \right) - v\left(\cup_{j \in S_h(\sigma_C)} C_j \cup S_i(\sigma_h) \right) \right]} }, \label{owen_def}
\ee
where $P_C$ is the set of permutations of $1, \ldots, m$; $S_h(\sigma_C)$ is the set of entries before $h$ in permutation $\sigma_C$; $P_h$ is the set of permutations of the participant indices in $C_h$, and $S_i(\sigma_h)$ is the set of indices before $i$ in $\sigma_h$. Continuing the jigsaw puzzle example, if each person helping put together the puzzle first attaches a set of pieces to each other to form a region of the puzzle, then the regions are joined together, then the Owen value for a piece is the expectation, over permutations of the regions paired with permutations of the pieces within that piece's region, of the marginal value produced if that piece is added to the combination of the regions before its region in the inter-region permutation and the pieces before it in the intra-region permutation. The Owen value for the coalition is the sum of the Owen values over its participants. In the puzzle case, the Owen value for each person putting together the puzzle is the sum of the Owen values for the pieces in their region.

\section{Some Useful Lemmas on Permutations} \label{review_perm}
Since Shapley values and Owen values are expectations over permutations, computing them efficiently will require gathering terms over permutations. We will use the following two lemmas to prove results about gathering those terms. (A note on notation: we will say "a permutation of $S$," where $S$ is a set, to mean a permutation of the elements of $S$.)

\begin{lemma} \label{lemma_one}
For $T \subseteq S$, each permutation of $T$ is the ordering of the elements of $T$ in an equal number of the permutations of $S$.
\end{lemma}

\begin{proof}
For each pair of permutations of $T$, there is a one-to-one mapping between permutations of $S$ with the elements of $T$ ordered according to the first permutation in the pair and those with the elements of $T$ ordered according to the second: remove the elements of $T$, then re-insert them into the same set of positions, but in the order given by the second permutation in the pair.
\end{proof}

\begin{lemma} \label{lemma_two}
Let $S_1, \ldots, S_m$ be $m$ disjoint sets, none of which contain element $i$. Let $t = |S_1 \cup \ldots \cup S_m \cup \si|$. Let $s_1, \ldots, s_m$ be integers with $0 \leq s_1 \leq |S_1|, \ldots, 0 \leq s_m \leq |S_m|$. Let $u = s_1 + \ldots + s_m$. Then the probability that a permutation drawn uniformly at random from the permutations of $S_1 \cup \ldots \cup S_m \cup \si$ has element $i$ preceded by exactly $s_1$ elements of $S_1$, \ldots, and $s_m$ elements of $S_m$ is
$$ \frac{1}{t} {{t-1}\choose{u}}^{-1} {{|S_1|}\choose{s_1}} \cdots {{|S_m|}\choose{s_m}}. $$
\end{lemma}

\begin{proof}
For a specific set of $u$ elements, consisting of $s_1$ from $S_1$, \ldots, and $s_m$ from $S_m$, there are $u!$ ways to order them in the first $u$ positions of a permutation, a single ($1!$) way to place element $i$ in position $u+1$, and $(t-(u+1))!$ ways to order the remaining elements in the last $t-(u+1)$ positions. So the probability of that specific set of $u$ elements preceding $i$ is
$$ \frac{u! 1! (t-(u+1))!}{t!} = \frac{1}{t} \frac{u! (t-1-u)!}{(t-1)!} = \frac{1}{t} {{t-1}\choose{u}}^{-1}. $$
The number of specific sets of $u$ elements consisting of $s_1$ from $S_1$, \ldots, $s_m$ from $S_m$ is
$$  {{|S_1|}\choose{s_1}} \cdots {{|S_m|}\choose{s_m}}. $$
\end{proof}

\begin{lemma} \label{lemma_three}
For $T \subset S$ and any element in $S - T$, equal numbers of the permutations of $S$ have the element preceded by $0, 1, \ldots, |T|$ elements of $|T|$.
\end{lemma}

\begin{proof}
Let $e$ be any element in $S - T$. Among the permutations of $T \cup \se$, $e$ is in each position $1$ to $|T| + 1$ for an equal number of permutations, because with $e$ in a position, there are $|T|!$ ways to arrange the elements of $T$ in the remaining $|T|$ positions. But the position of $e$ in a permutation of $T \cup \se$ is one greater than the number of elements that precede $e$, so that number is each of $0$ to $|T|$ for an equal number of permutations of $T \cup \se$. It only remains to show that each permutation of $T \cup \se$ is the ordering of its elements in an equal number of permutations of $S$. But that is Lemma \ref{lemma_one}, with $T \cup \se$ as $T$. 
\end{proof}

The first lemma implies that if we select a permutation of $S$ at random, then the ordering of the elements of $T$ within the permutation is equally likely to be each permutation of $T$. The second lemma offers a way to compute the probability over permutations that a specific element is preceded by specified numbers of elements from different sets. The third lemma implies that we can build a random permutation of $S$ by starting with a single element, then, for each remaining element, select an insert position uniformly at random from: before the first entry, between the first and second, between the second and third, ..., after the last. 

A note on notation before we proceed to the main results: for combinations, let ${{a}\choose{b}}$ be zero if $b<0$ or $b>a$. This will avoid the need for complex bound-checking in summations that have combinations in their terms.

\section{Frequency-Based Decision Making} \label{freq_based}
Assume that there is a set of in-sample examples, each with an input and a label. Then, for each of a set or series of out-of-sample inputs, the in-sample examples are used to generate outputs, based on statistics or machine learning. Each output, which may be a classification, decision, or score, leads to some action, which has some value. The value of an action may be positive, zero, or negative. In general, the value of an action is not known when the action is taken. However, we assume that each action's value can be assessed later, when valuing contributions of data providers. We also assume that we can assess values for actions not taken.

By \textit{frequency-based decision-making}, we mean a process that, when supplied with an out-of-sample input, identifies which of a set of pre-defined bins contains it, then generates an output based solely on the frequencies of labels over in-sample examples in that bin. In other words, if two different sets of in-sample examples have the same frequencies over labels in a bin, then they generate the same outputs for all out-of-sample inputs in that bin.

Perhaps the simplest frequency-based decision-making process is to have a single, universal bin and label each out-of-sample input with the most common label among the in-sample examples. More complex bin schemes include histograms \cite{devroye96} (Section 6.4) and trees to divide the input space into multiple bins in different ways. More complex decision-making methods or scoring functions abound. They include producing no label unless the frequencies meet some standard of statistical significance, producing a lower or upper confidence bound on the probability of a label given the frequencies, producing a Bayesian estimate of the probability of a label given the frequencies and a prior, producing a smoothed output such as a sigmoid applied to the relative frequency of a label, or producing the output of an estimation procedure based on the relative frequency and the number of samples. (For details on estimation procedures, refer to \cite{wilson27,clopper34,agresti98,brown01,ross03}.) 

\subsection{Shapley Values for Examples}
Let $x$ be an out-of-sample input, and let $y$ be its label. (Note that $y$ is unknown at decision and action time, but becomes known before we allocate value to the in-sample examples.) We will derive Shapley values for each in-sample example in the same bin as $x$. The value for examples outside the bin is zero. To produce a Shapley value for an in-sample example used for multiple out-of-sample examples, sum the Shapley values over the out-of-sample examples. 

Let $n$ be the number of in-sample examples in the same bin as $x$. Let $1, \ldots, n$ index the in-sample examples in the same bin as $x$. For simplicity and clarity, focus on binary classification. (For more than two label values, the concepts are the same, but the notation is unwieldy.) For each example $i \in \{1, \ldots, n\}$, let $a(i) = 1$ and $b(i) = 0$ if the label of example $i$ equals $y$, and $a(i) = 0$ and $b(i) = 1$ otherwise. Let $A$ be the set of indices of in-sample examples in the same bin as $x$, excluding example $i$, that have label equal to $y$ (i.e. the same label as $x$). Let $B$ be the set of indices of examples in the same bin as $x$, excluding example $i$, that have a different label than $y$. 
 
Let $v(a, b)$ be the value of the action based on $x$, given the output that would result if the in-sample examples in the same bin as $x$ included exactly $a$ examples that have label $y$ and $b$ examples that do not. Let $\Delta_i v(a, b)$ be the change in value from adding example $i$ to the bin, given that $|A| = a$ and $|B| = b$:
 $$ \Delta_i v(a, b) = v(a + a(i), b + b(i)) - v(a, b). $$
 Let $R$ be the set of $(a,b)$ pairs for which adding example $i$ changes the value:
 $$ R = \{(a,b) |  \Delta_i v(a, b) \not= 0\}. $$
 Call $R$ the critical set. 
 
 \begin{theorem} \label{shapley_dec}
The Shapley value for example $i$ on out-of-sample input $x$ with label $y$ is:
 $$ \sum_{(a,b) \in R} \frac{1}{n} {{n - 1}\choose{a+b}}^{-1} {{|A|}\choose{a}} {{|B|}\choose{b}} \hbox{ } \Delta_i v(a, b). $$
 \end{theorem}
 
 \begin{proof}
Begin with Expression \ref{shapley_def}, the general equation for a Shapley value from Section \ref{review_shapley}. The values $v(S_i(\sigma))$ and $v(S_i(\sigma) \cup \si)$ are only affected by in-sample examples in the same bin as $x$. By Lemma \ref{lemma_one}, the same number of permutations of all examples contain each ordering of those $n$ examples.  So average over the $n!$ permutations of those examples instead of the permutations of all in-sample examples. By Lemma \ref{lemma_two} the probability that $a$ indices from $A$ and $b$ from $B$ come before $i$ is
$$ \frac{1}{n} {{n - 1}\choose{a+b}}^{-1} {{|A|}\choose{a}} {{|B|}\choose{b}}. $$
Given that this occurs, 
$$ \Delta_i v(a, b) = v(S_i(\sigma) \cup \si) - v(S_i(\sigma)). $$
Restricting the sum to terms with $(a,b) \in R$ only avoids terms with $\Delta_i v(a, b) = 0$, so it has no impact on the sum.
 \end{proof}
 
As a simple example, suppose we use a single bin that includes all in-sample examples, a classification rule that outputs the majority label and has no output if there is no majority, and a value function that is 100 if the classification is correct, -500 if incorrect, and zero if there is no output. Suppose example $i$ has label $y$. If $a = b$, then adding example $i$ to a bin with $a$ examples with label $y$ and $b$ without turns no output into a correct classification. Similarly, if $a = b - 1$, then adding example $i$ changes from incorrect classification to no output. So $R = \{(a,b) | a = b\} \cup \{(a,b) | a = b - 1\}$, and the Shapley value for example $i$ is:
$$ \frac{100}{n} \sum_{\{(a,b) | a = b\}} {{n - 1}\choose{a+b}}^{-1} {{|A|}\choose{a}} {{|B|}\choose{b}} + \frac{500}{n} \sum_{\{(a,b) | a = b - 1\}} {{n - 1}\choose{a+b}}^{-1} {{|A|}\choose{a}} {{|B|}\choose{b}}. $$
Alternatively, if example $i$ does not have label $y$, then adding it creates a tie from a correct classification if $a = b + 1$, for a loss of 100, and creates an incorrect classification from a tie if $a = b$, for a loss of 500. So the Shapley value is:
$$ -\frac{100}{n} \sum_{\{(a,b) | a = b + 1\}} {{n - 1}\choose{a+b}}^{-1} {{|A|}\choose{a}} {{|B|}\choose{b}} - \frac{500}{n} \sum_{\{(a,b) | a = b\}} {{n - 1}\choose{a+b}}^{-1} {{|A|}\choose{a}} {{|B|}\choose{b}}. $$

The Shapley values are the same for all in-sample examples in the same bin and with the same label. Also, the Shapley values for each in-sample example are the same for all out-of-sample examples in the same bin and with the same label. So computing the Shapley values for all in-sample examples over a set of out-of-sample examples requires only computing the Shapley values for any in-sample example in each bin with each label for an out-of-sample example in each bin with each label -- four different Shapley value computations per bin, one for each combination of in-sample label and out-of-sample label.

\subsection{Owen Values for Sets of Examples}
Suppose that a set of coalitions (for example, data-gathering organizations) each contribute subsets of the examples for the in-sample set. Let $m$ be the number of coalitions, and let $C_1, \ldots, C_m$ be the sets of indices of the (within-bin) in-sample examples contributed by each coalition. We will compute Owen values for each in-sample example, for the value of a decision and action on an out-of-sample example $x$ with label $y$. (The label is unknown at decision time, but we assume that it is known at the time that data values are computed.) The Owen value for an in-sample example over a set of out-of-sample examples is the sum of Owen values for each out-of-sample example. The Owen value for a coalition is the sum of Owen values over the in-sample examples that the coalition contributes. 

Owen values are nested expectations, over permutations of coalitions and permutations of in-sample examples within each coalition. We will use dynamic programming to account for permutations over coalitions, then use the results from the previous section for permutations over examples within each coalition. To make the notation simpler, assume that we are computing the Owen value for an example in coalition $m$, the last coalition. (For the others, we can simply reorder the indices of the coalitions to make the coalition of interest become coalition $m$.) 

Let $p_{j,s,a,b}$ be the probability that a random permutation of the coalitions places $s$ of the coalitions $C_1, \ldots, C_j$ before $C_m$ and the others after, and those $s$ coalitions collectively contribute $a$ in-samples examples with label $y$ and $b$ examples with the other label. Then, by definition, base case values are:
$$ p_{0,0,0,0} = 1 $$
and 
$$ \forall (s, a, b) \not= (0,0,0): p_{0,s,a,b} = 0. $$
And the following recurrence holds:

\begin{lemma} \label{lemma_rec}
Let $a_j$ be the number of in-sample examples indexed by $C_j$ that have label $y$ and let $b_j$ be the number that do not. Then
$$ p_{j,s,a,b} = \frac{s+1}{j+1} p_{j-1,s-1,a-a_j,b-b_j} + \frac{j-s}{j+1} p_{j-1,s,a,b}. $$
\end{lemma}

\begin{proof}
The first term accounts for $C_j$ preceding $C_m$ in the permutation of coalitions, and the second term accounts for $C_j$ following $C_m$. From Lemma \ref{lemma_three}, among all permutations of $C_1, \ldots, C_m$, $C_j$ is equally likely to be preceded by each number from $0$ to $j+1$ of the coalitions $C_1, \ldots C_{j-1}$ and $C_m$. If $s$ of $C_1, \ldots, C_{j-1}$ come before $C_m$, then $C_j$ comes before $C_m$ too if $s$ or fewer of $C_1, \ldots, C_{j-1}$ and $C_m$ precede $C_j$. From $0$ to $s$ is $s+1$ possibilities, so the probability that $C_j$ comes before $C_m$ is $\frac{s+1}{j+1}$. The probability that it comes after is $1 - \frac{s+1}{j+1} = \frac{j-s}{j+1}$. 
\end{proof}

\begin{theorem}
Let 
$$ p_{a,b} = \sum_{s=0}^{m-1} p_{m-1, s, a, b}. $$
Then the Owen value for in-sample example $i$ contributed by coalition $C_m$ on out-of-sample input $x$ with label $y$ is:
$$ \sum_{(a,b) \in R} \sum_{(a',b') | a' \leq a, b' \leq b} p_{a - a', b - b'} \hbox{ } \frac{1}{|C_m|} {{|C_m| - 1}\choose{|A_m| + |B_m|}}^{-1} {{|A_m|}\choose{a'}} {{|B_m|}\choose{b'}} \hbox{ } \Delta_i(a,b), $$
where $A_m$ is the set of examples indexed by $C_m$, excluding example $i$, that have label $y$, and $B_m$ is the set that do not.
\end{theorem}

\begin{proof}
Recall the definition of the Owen value from Expression \ref{owen_def} in Section \ref{review_shapley} (with $h = m$ for us). By Lemma \ref{lemma_rec}, $p_{a - a', b - b'}$ is the probability over permutations of coalitions that the coalitions preceding $C_m$ in a permutation contribute $a - a'$ within-bin in-sample examples with label $y$ and $b - b'$ without. By Lemma \ref{lemma_one} we can take the inner average over only permutations of the elements of $C_m$, ignoring any out-of-bin examples. By Lemma \ref{lemma_two},  the probability that the indices of $a'$ of those examples with label $y$, and $b'$ without, precede $i$ is
$$ \frac{1}{|C_m|} {{|C_m| - 1}\choose{|A_m| + |B_m|}}^{-1} {{|A_m|}\choose{a'}} {{|B_m|}\choose{b'}}. $$
So the product is the probability that the preceding coalitions and examples from $C_m$ preceding $i$ together contribute $a$ examples with label $y$ and $b$ without. For expectation, we need probability times value, and $\Delta_i(a,b)$ is the marginal value from adding example $i$. Finally, the summation over $R$ only leaves out zero-value terms. 
\end{proof}

Within each coalition, in-sample examples in the same bin and with the same label have the same Owen value. Also, for each in-sample example, the Owen values are the same for all out-of-sample examples in the same bin and with the same label. So computing Owen values for all in-sample examples over a set of out-of-sample examples requires only computing the Owen value for each combination of coalition, bin, in-sample label, and out-of-sample label.

\section{Nearest-Neighbor Classifiers} \label{knn}
To classify each out-of-sample input $x$, a $k$-nearest neighbor ($k$-nn) classifier \cite{cover67,cover68,duda01,devroye96} first identifies the $k$ nearest neighbors, which are the in-sample examples with inputs closest to $x$ according to some metric. Then the classifier outputs the label shared by the majority of the $k$ nearest neighbors. We assume $k$ is odd and binary classification, so no ties in voting. Also, we assume the metric has consistent tie-breaking, in order to have consistent neighbors for the same $x$. To make a metric meet this condition with probability one, augment each input with a real number, and in case of a distance tie, use the distance between those real numbers to settle the tie \cite{devroye79}. Note that the metric can be any function that takes two example inputs and returns a number -- the metric need not follow the triangle inequality and need not be symmetric. 

\subsection{Shapley Values for $k$-Nearest Neighbors}
Let $n$ be the number of in-sample examples. Refer to in-sample examples by indices 1 to $n$. To compute the Shapley value for example $i$ in the classification of out-of-sample input $x$ with actual label $y$, note that adding example $i$ to a set of in-sample examples indexed by $S$ can affect the value in two different ways: (1) if $S$ has $k-1$ examples, then adding example $i$ makes a vote among $k$ examples possible, and (2) if $S$ has $k$ or more examples, then adding example $i$ may displace one of the $k$ nearest neighbors from $S$, which may change the vote. We will handle these cases separately.  

For case (1), assume that the $k$-nn classifier makes no decision if it has fewer than $k$ in-sample examples. Let $v_n$ be the value if there is no decision. Let $v_c$ be the value of a correct classification, and let $v_w$ be the value of an incorrect one. If $|S| = k-1$, then $v(S) = v_n$. Also, $v(S \cup \si) = v_c$ if the majority of the labels on examples indexed by $S \cup \si$ are $y$, and $v(S \cup \si) = v_w$ otherwise. 

\begin{lemma} \label{create_shapley}
Let $y_i$ be the label of example $i$, let indicator function $I()$ be one if its argument is true and zero otherwise, and let $A$ index the in-sample examples, excluding example $i$, that have label $y$. Then the contribution to the Shapley value for example $i$ creating a classification of out-of-sample input $x$ with label $y$ is:
$$ f_i(x,y) = \frac{1}{n} {{n-1}\choose{k-1}}^{-1} \left[ \sum_{a=0}^{\frac{k-1}{2} - I(y_i=y)} {{|A|}\choose{a}} {{n - 1 - |A|}\choose{k - 1 - a}} v_w + \sum_{a = \frac{k-1}{2} - I(y_i=y) + 1}^{k-1} {{|A|}\choose{a}} {{n - 1 - |A|}\choose{k - 1 - a}} v_c \right] - \frac{v_n}{n}. $$
\end{lemma}

\begin{proof}
The fraction of permutations with $i$ in position $k$ is $\frac{1}{n}$, by symmetry. For those permutations, each set of $k-1$ of the $n-1$ other examples is equally common as the first $k-1$ examples in the permutation. Sets of $k-1$ examples with $\frac{k-1}{2} - I(y_i=y)$ or fewer examples with label $y$ have a majority of labels not equal to $y$ once example $i$ is included to make $k$ votes, so $v(S \cup \si) = v_w$. Those with more examples with label $y$ result in correct classification, so $v(S \cup \si) = v_c$. In both cases, we must subtract $v(S) = v_n$.
\end{proof}
 
For case (2), let example $j$ be the $k$th nearest neighbor to $x$ in $S$. (We will say an example is in $S$ to denote that its index is in $S$.)  If example $i$ is closer to $x$ than example $j$ is, then example $i$ displaces example $j$ as a voter, which may alter the classifier's decision. Adding $i$ to $S$ changes the classification if and only if all these conditions hold:
\begin{itemize}
\item Example $i$ is nearer to $x$ than example $j$ is, so that example $i$ replaces example $j$ as a voter.
\item Example $i$ has a different label than example $j$, so they vote differently. 
\item Exactly half the $k-1$ nearest neighbors to $x$ in $S$ have label $y$, so example $i$ changes the majority vote. 
\end{itemize}

If adding example $i$ meets these conditions and example $i$ has label $y$, then $v(S \cup \si) - v(S) = v_c - v_w$, because adding $i$ corrects an incorrect classification. If it meets the conditions and does not have label $y$, then $v(S \cup \si) - v(S) = v_w - v_c$. 

\begin{lemma} \label{change_shapley}
Let $J$ index the in-sample examples that have a different label than example $i$ and a greater distance to $x$. If example $i$ has label $y$, then let $\Delta_i v = v_c - v_w$; otherwise let $\Delta_i v = v_w - v_c$. Let $A_j$ index the in-sample examples, excluding example $i$, that are closer to $x$ than example $j$ and have label $y$. Let $B_j$ index those that do not have label $y$. Then the Shapley value for example $i$ changing the classification of out-of-sample input $x$ with label $y$ is:
$$ g_i(x,y) = \sum_{j \in J} \frac{1}{|A_j| + |B_j| + 2} {{|A_j| + |B_j| + 1}\choose{k}}^{-1} {{|A_j|}\choose{\frac{k-1}{2}}} {{|B_j|}\choose{\frac{k-1}{2}}} \Delta_i v. $$
\end{lemma}

\begin{proof}
Summation over $J$ is needed to ensure the first two conditions for example $i$ to change the outcome of the vote. The marginal value $\Delta_i v$ is the correct value for changing the vote. By Lemma \ref{lemma_one}, for each $j$, we can average over permutations of just $A_j \cup B_j \cup \{j\} \cup \{i\}$ rather than all in-sample examples. For the probability that a permutation has exactly $\frac{k-1}{2}$ of the indices in $A_j$, exactly $\frac{k-1}{2}$ of the indices in $B_j$, and $j$ all preceding $i$, apply Lemma \ref{lemma_two}, with $S_1 = A_j$, $S_2 = B_j$, and $S_3 = \{j\}$: 
$$ \frac{1}{|A_j| + |B_j| + 2} {{|A_j| + |B_j| + 1}\choose{k}}^{-1} {{|A_j|}\choose{\frac{k-1}{2}}} {{|B_j|}\choose{\frac{k-1}{2}}} {{1}\choose{1}}. $$
$$ = \frac{1}{|A_j| + |B_j| + 2} {{|A_j| + |B_j| + 1}\choose{k}}^{-1} {{|A_j|}\choose{\frac{k-1}{2}}} {{|B_j|}\choose{\frac{k-1}{2}}}. $$
\end{proof}

\begin{theorem} \label{knn_shapley}
The Shapley value for example $i$ in the classification of out-of-sample input $x$ with label $y$ is:
$$ f_i(x,y) + g_i(x,y). $$
\end{theorem}

\begin{proof}
Combine Lemmas \ref{create_shapley} and \ref{change_shapley}.
\end{proof} 

Now consider how to compute the formula in Theorem \ref{knn_shapley} for all in-sample examples. The terms of $f_i(x,y)$ are easy to compute, because they have the same values for all examples that have the same label, they are sums over $k$ or fewer terms, and the only computation over multiple examples is to determine how many examples have the same label as $y$. 

For $g_i(x,y)$, to simplify notation, assume that examples are numbered in order of how near they are to $x$. In other words, assume example 1 is the closest to $x$, example 2 is the second closest, and so forth. Let $a_j$ be the number of examples closer to $x$ than example $j$ that have label $y$. Then
$$ |A_j| = a_j - I(y_i = y), $$
because $a_j$ excludes example $i$. Similarly, let $b_j$ be the number of examples closer to $x$ than example $j$ that do not have label $y$. Then
$$ |B_j| = b_j - I(y_i \not= y). $$
Since $y_i = y$ or $y_i \not= y$,
$$ |A_j| + |B_j| = a_j + b_j - 1. $$

For each out-of-sample example $(x,y)$, compute values $a_j$ and $b_j$ using the recurrence:
$$ a_1 = b_1 = 0, $$
$$ a_j = a_{j-1} + I(y_{j-1} = y), $$
and
$$ b_j = b_{j-1} + I(y_{j-1} \not= y). $$

Note that $g_i(x,y)$ is equal to
$$ s_i \equiv \sum_{j=i+1}^{n} I(y_j \not= y_i) \frac{1}{a_j + b_j + 1} {{a_j + b_j}\choose{k}}^{-1} {{a_j - I(y_i = y)}\choose{\frac{k-1}{2}}} {{b_j - I(y_i \not= y)}\choose{\frac{k-1}{2}}} \Delta_i v. $$
For each label value $u$, let
$$ s_{i,u} \equiv \sum_{j=i+1}^{n} I(y_j \not= u) \frac{1}{a_j + b_j + 1} {{a_j + b_j}\choose{k}}^{-1} {{a_j - I(u = y)}\choose{\frac{k-1}{2}}} {{b_j - I(u \not= y)}\choose{\frac{k-1}{2}}} \Delta_i v. $$
Then $s_i = s_{i,y_i}$. To compute $s_{i,u}$ for each value of $u$, use the recurrence:
$$ s_{n,u} = 0, $$
and
$$ s_{i,u} = s_{i+1,u} + I(y_{i+1} \not= u) \frac{1}{a_{i+1} + b_{i+1} + 1} {{a_{i+1} + b_{i+1}}\choose{k}}^{-1} {{a_{i+1} - I(u = y)}\choose{\frac{k-1}{2}}} {{b_{i+1} - I(u \not= y)}\choose{\frac{k-1}{2}}} \Delta_i v. $$
To see that the recurrence holds, notice from the definition of $s_i$ that $s_i$ is just $s_{i+1}$ with one more term in the sum -- the term with $j = i+1$. Using this recurrence, Shapley values for each out-of-sample example can be computed for all in-sample examples in O($n$) time. If we include time to sort the examples in order of distance from $x$, then the time becomes O($n \lg n$)


\subsection{Owen Values for $k$-Nearest Neighbors}
Assume that $C_1, \ldots, C_m$ index the in-sample examples contributed by $m$ coalitions to form the full set of in-sample examples. We will present a method to compute the Owen value for in-sample example $i$ in coalition $m$ (to make the notation simpler, and without loss of generality, since we can renumber the coalitions) for the action taken based on $k$-nearest neighbor classification of an out-of-sample input $x$ with label $y$. To compute the Owen value for a coalition, sum the Owen values over its examples. 

Begin with the portion of the Owen values for the addition of example $i$ changing the decision, which corresponds to $g_i(x,y)$. For each example $j$ indexed by $J$, let $q_{h, s, a, b}$ be the probability over permutations of the coalitions that $s$ of the coalitions among $C_1, \ldots, C_h$ precede $C_m$ and those coalitions together contribute $a$ examples indexed by $A_j$, $b$ examples indexed by $B_j$ and, if $j \not\in C_m$, example $j$. Without loss of generality, if $j \not\in C_m$ then let $j \in C_1$. (Renumber coalitions if needed.) If $j \in C_m$ then base cases are:
$$ q_{0,0,0,0} = 1 \hbox{ and } \forall (s, a, b) \not = (0,0,0): q_{0, s, a, b} = 0. $$
Let $a_h = |A_j \cap C_h|$. Let $b_h = |B_j \cap C_h|$. If $j \not \in C_m$ then base cases are:
$$ q_{1, 1, a_1, b_1} = \frac{1}{2} \hbox{ and } \forall (s, a, b) \not= (1,  a_1, b_1): q_{1, s, a, b} = 0. $$
For the recurrence:

\begin{lemma} \label{lemma_rec_knn}
$$ q_{h, s, a, b} = \frac{s+1}{h+1} q_{h-1, s-1, a-a_h, b-b_h} + \frac{h - s}{h + 1} q_{h-1, s, a, b}. $$
\end{lemma}

\begin{proof}
The recurrence is very similar to the one in Lemma \ref{lemma_rec}. To prove it, just replace $p$ by $q$ and $j$ by $h$ in the proof of Lemma \ref{lemma_rec}.
\end{proof}

\begin{lemma} \label{knn_owen}
Let 
$$ q_{a,b} = \sum_{s=0}^{m-1} q_{m-1, s, a, b}. $$
Let $J_m = J \cap C_m$. Then the portion of the Owen value for in-sample example $i$ changing the $k$-nearest neighbor classification decision for out-of-sample input $x$ with label $y$, given that example $i$ is contributed by coalition $C_m$, is:
$$ \hat{g}_{m,i}(x,y) = \sum_{j \in J - J_m} \sum_{(a,b)|a \leq a_m, b \leq b_m} q_{\frac{k-1}{2} - a, \frac{k-1}{2} - b} \hbox{ } \frac{1}{a_m + b_m + 1} {{a_m + b_m}\choose{a + b}}^{-1} {{a_m}\choose{a}} {{b_m}\choose{b}} \hbox{ } \Delta_i v, $$
$$ + \sum_{j \in J_m} \sum_{(a,b)|a \leq a_m, b \leq b_m} q_{\frac{k-1}{2} - a, \frac{k-1}{2} - b)} \hbox{ } \frac{1}{a_m + b_m + 2} {{a_m + b_m + 1}\choose{a + b + 1}}^{-1} {{a_m}\choose{a}} {{b_m}\choose{b}} \hbox{ } \Delta_i v. $$
\end{lemma}

\begin{proof}
Together, the sums cover summation over $J$, ensuring that adding example $i$ can change the vote. In both sums, $q_{\frac{k-1}{2} - a, \frac{k-1}{2} - b}$ supplies the probability that coalitions before $C_m$ in the coalition permutation supply $\frac{k-1}{2} - a$ examples from $A_j$ and $\frac{k-1}{2} - b$ from $B_j$. For the first sum, a previous coalition supplies $j$, so $C_m$ only needs to supply $a$ of $a_m$ and $b$ of $b_m$ elements before element $i$ in the coalition $C_m$ permutation in order to have a tied vote without $j$. By Lemma \ref{lemma_one}, only probabilities over permutations of $(A_j \cap C_m) \cup (B_j \cap C_m) \cup \{i\}$ need to be considered. By Lemma \ref{lemma_two}, the probability that $C_m$ contributions before example $i$ create a tie without $j$ is
$$ \frac{1}{a_m + b_m + 1} {{a_m + b_m}\choose{a + b}}^{-1} {{a_m}\choose{a}} {{b_m}\choose{b}}. $$
For the second sum, $j$ must come before $i$ in the coalition $C_m$ permutation, and $a$ of $a_m$ and $b$ of $b_m$ must also. By Lemma \ref{lemma_one}, only permutations of $(A_j \cap C_m) \cup (B_j \cap C_m) \cup \{j\} \cup \{i\}$ need to be considered. By Lemma \ref{lemma_two}, with $S_1 = A_j \cap C_m$, $S_2 =  B_j \cap C_m$, and $S_3 = \{j\}$, the probability that $C_m$ contributions before example $i$, not including $j$, create a tie, and $j$ also precedes $i$, is
$$ \frac{1}{a_m + b_m + 2} {{a_m + b_m + 1}\choose{a + b + 1}}^{-1} {{a_m}\choose{a}} {{b_m}\choose{b}} {{1}\choose{1}} $$
$$ = \frac{1}{a_m + b_m + 2} {{a_m + b_m + 1}\choose{a + b + 1}}^{-1} {{a_m}\choose{a}} {{b_m}\choose{b}}. $$
\end{proof}

Now consider the portion of the Owen value for an example $i$ from $C_m$ creating a classification by being the $k$th example in the classifier. (This portion corresponds to $f_i(x,y)$.) As before, let $A$ index the in-sample examples, excluding example $i$, that have label $y$. Let $B$ index those that do not have label $y$. Reuse some notation, with $A$ and $B$ in place of $A_j$ and $B_j$: let
$$ a_h = |A \cap C_h| $$
and
$$ b_h = |B \cap C_h|. $$
Use these values to compute $q_{h, s, a, b}$ and $q_{a,b}$ using the formulas given previously. Then $q_{a,b}$ is the probability over permutations of coalitions that the coalitions preceding $C_m$ together have $a$ examples with label $y$ and $b$ without. 

\begin{lemma} \label{knn_owen_create}
Using these $a_h$, $b_h$ and $q_{a,b}$ values, the portion of the Owen value for example $i$ from $C_m$ creating a classification decision by being the $k$th example in a $k$-nn classifier is 
$$ \hat{f}_{m,i}(x,y) = \sum_{(a,b,a',b')|a+b+a'+b'=k-1 \hbox{ and } a + a' + I(y_i=y) < \frac{k}{2}} q_{a,b} \frac{1}{|C_m|} {{|C_m|-1}\choose{a'+b'}}^{-1} {{a_m}\choose{a'}} {{b_m}\choose{b'}} (v_w - v_n) $$
$$ + \sum_{(a,b,a',b')|a+b+a'+b'=k-1 \hbox{ and } a + a' + I(y_i=y) > \frac{k}{2}} q_{a,b} \frac{1}{|C_m|} {{|C_m|-1}\choose{a'+b'}}^{-1} {{a_m}\choose{a'}} {{b_m}\choose{b'}} (v_c - v_n). $$
\end{lemma}

\begin{proof}
In each term, $a$ is the number of examples with label $y$, and $b$ is the number without, contributed by coalitions preceding $C_m$, and $a'$ and $b'$ are the numbers of examples with and without label $y$ among examples in $C_m$ that precede example $i$ in a permutation of the examples in $C_m$. Both sums require $k-1$ examples preceding example $i$ by requiring $a+b+a'+b'=k-1$. The first sum requires incorrect classification ($a + a' + I(y_i=y) < \frac{k}{2}$) and has value term $v_w - v_n$, for the net value from no decision to an incorrect one. The second sum requires correct classification ($a + a' + I(y_i=y) > \frac{k}{2}$) and has value term $v_c - v_n$, for the net value of creating a correct decision from no decision. In both, sums, $q_{a,b}$ is the probability that coalitions other than $C_m$ contribute $a$ correct and $b$ incorrect votes, $\frac{1}{|C_m|}$ is the probability over permutations of examples in $C_m$ that example $i$ is in position $k - (a+b)$, making it the $k$th example in the classifier, and the other terms are the probability that the $(k-1) - (a+b) = a'+b'$ examples preceding example $i$ are $a'$ with correct labels and $b'$ with incorrect labels. 
\end{proof}

\begin{theorem}
The Owen value for example $i$ in coalition $C_m$ in the classification of out-of-sample input $x$ with label $y$ is: 
$$ \hat{f}_{m,i}(x,y) + \hat{g}_{m,i}(x,y). $$
\end{theorem}

\begin{proof}
Combine Lemmas \ref{knn_owen} and \ref{knn_owen_create}. 
\end{proof}

\section{Discussion} \label{discussion}
This paper provides methods to divide profits or losses from data among data providers. In the future, it would be useful to expand the division of profits or losses to include the party or parties that perform the analysis and take action based on it to realize a profit or loss. One possible approach is to include the analyzer and administrator as a participant or participants in the coalition when computing Shapley or Owen values. Other, more prosaic, approaches include applying Shapley or Owen values only to ``excess'' profits, perhaps as defined by an a priori agreement between the data providers and an analyzer/administrator, for example having each user pay a fixed fee for analysis and administration, then dividing the remaining profit or loss among the data providers -- similar to the fees paid to money managers in mutual funds. In such cases, the gain or loss from a decision for the purposes of computing the data dividend should be limited to the portion of the gain or loss that contributes to the excess profits only. Each decision's contribution to the excess profits could be computed using Shapley or Owen values. 

In a competitive environment such as investing, making data available to multiple decision-makers may decrease the value of the data to each. Conversely, if one organization can deny data to others, it may gain a competitive advantage. It would be interesting to examine how such strategic interests might influence valuation and payment for data. For example, how much should a decision-maker value having exclusive access to data? Conversely, for some applications, data may have decreasing marginal value for decision-making because a limited amount may allow statistically significant estimates and then more data would be unlikely to change decisions. So statistical sufficiency and competitive pressures may balance each other in determining prices for data. 

With payment for data comes an incentive to generate data that is false or duplicated, perhaps with slight changes to avoid detection. So paying for data may prompt a need for data verification. This is important not just to avoid over-paying those who contribute such data, but also to avoid making suboptimal decisions because of it. 

Finally, some people's data may be more valuable than others' to some organizations or for some functions. For example, data about a person whose spending behavior is indicative of a larger group's may offer more useful insights on which products are having increased sales than data about someone who has more unusual spending habits. However, there may be more people who have common spending habits, making data about them less scarce. In either case, people who are representative of groups who spend more seem likely to have more valuable data. As a result, offering to pay for data based on its value may have the potential to increase inequality. (For more on big data and inequality, refer to \cite{oneil16}.)

Groups of people who are more protective of their privacy may contribute less data, and this could be exacerbated by paying for data. To get a representative sample, it may be necessary to pay people from those groups more for their data. This should help reduce some forms of bias in the analysis of the data and in decisions based on it. (For more on bias and the web, refer to \cite{baeza-yates18}.)


\pagebreak

\bibliographystyle{ACM-Reference-Format}
\bibliography{bax} 

\end{document}